\documentclass[12pt,draftcls,peerreview, onecolumn]{IEEEtran}

\usepackage{times, epsf}
\usepackage{amsmath,amssymb}
\usepackage{graphicx}
\usepackage{color}
\usepackage{subfigure}
\usepackage{cite}
\usepackage{multirow,tabularx}

\usepackage{wasysym}
\usepackage{stackrel}

\newtheorem{theorem}{{\bf Theorem}}
\newtheorem{lemma}{{\bf Lemma}}

\newtheorem{corollary}{{\bf Corollary}}

\newcommand{\qed}{\nobreak \ifvmode \relax \else
      \ifdim\lastskip<1.5em \hskip-\lastskip
      \hskip1.5em plus0em minus0.5em \fi \nobreak
      \vrule height0.75em width0.5em depth0.25em\fi}

\newcommand{\beq}{\begin{equation}}
\newcommand{\eeq}{\end{equation}}

\newcommand{\bdisp}{\begin{displaymath}}
\newcommand{\edisp}{\end{displaymath}}

\newcommand{\beqarr}{\begin{eqnarray}}
\newcommand{\eeqarr}{\end{eqnarray}}

\newcommand{\bmlt}{\begin{multline}}
\newcommand{\emlt}{\end{multline}}

\newcommand{\beqarrn}{\begin{eqnarray*}}
\newcommand{\eeqarrn}{\end{eqnarray*}}

\newcommand{\barr}{\begin{array}}
\newcommand{\earr}{\end{array}}

\newcommand{\benum}{\begin{enumerate}}
\newcommand{\eenum}{\end{enumerate}}

\newcommand{\bit}{\begin{itemize}}
\newcommand{\eit}{\end{itemize}}

\newcommand{\bc}{\begin{center}}
\newcommand{\ec}{\end{center}}

\newcommand{\bdes}{\begin{description}}
\newcommand{\edes}{\end{description}}

\newcommand{\bfig}{\begin{figure}}
\newcommand{\efig}{\end{figure}}

\newcommand{\bemq}{\begin{quote} \begin{em}}
\newcommand{\eemq}{\end{em} \end{quote}}

\newcommand{\bmp}{\begin{minipage}}
\newcommand{\emp}{\end{minipage}}

\newcommand{\eqn}[1]{(\ref{#1})}

\newcommand{\brac}[1]{\left({#1}\right)}

\newcommand{\cbrac}[1]{\left\{{#1}\right\}}

\newcommand{\floor}[1]{\left\lfloor{#1}\right\rfloor}
\newcommand{\ceil}[1]{\left\lceil {#1} \right\rceil}

\newcommand{\define}{\triangleq}

\newcommand{\ie}{{\it i.e.}}

\newcommand{\iid}{{i.i.d.}}

\newcommand{\expect}[1]{{\bf E}\left[{#1}\right]}

\newcommand{\prob}[1]{\text{Pr}\brac{#1}}

\newcommand{\bsp}{\begin{slide*}}
\newcommand{\esp}{\end{slide*}}
\newcommand{\bsl}{\begin{slide}}
\newcommand{\esl}{\end{slide}}

\IEEEoverridecommandlockouts

\newcommand{\tmax}{T_{\max}}

\newcommand{\Ps}[2]{P_{#1}}

\newcommand{\Psopt}[2]{P_{#1}^{*}}

\newcommand{\alp}[2]{\alpha_{#2}[#1]} 

\newcommand{\alpopt}[2]{\alpha_{#2}^{*}[#1]} 

\newcommand{\bet}[2]{\beta_{#2}[#1]} 

\newcommand{\betopt}[2]{\beta_{#2}^{*}[#1]} 

\newcommand{\avetime}[2]{\Gamma_{#1}}

\newcommand{\avetimeopt}[2]{\Gamma^{*}_{#1}}

\newcommand{\mapopt}{f^{*}}
\newcommand{\mapoptLag}{f^{\lambda *}}

\begin{document}

\title{Optimal Timer Based Selection Schemes}

\author{Virag Shah, {\it Student Member, IEEE}, Neelesh B.  Mehta, {\it Senior Member, IEEE},  Raymond Yim, {\it Member, IEEE} %
  \thanks{A part of this work is being submitted to ICC 2010.}
  \thanks{V.\ Shah and N.\ B.\ Mehta are with the Electrical
    Communication Engineering Dept. at the Indian Institute of Science
    (IISc), Bangalore, India.  R.\ Yim is with the Mitsubishi Electric
    Research Labs (MERL), Cambridge, MA, USA.}  \thanks{Emails: \{\tt
    virag4u@gmail.com, nbmehta@ece.iisc.ernet.in, yim@merl.com\}.}
 }

 \maketitle

\begin{abstract}
  
  Timer-based mechanisms are often used to help a given (sink) node
  select the best helper node among many available nodes.
  Specifically, a node transmits a packet when its timer expires, and
  the timer value is a monotone non-increasing function of its local
  suitability metric.  The best node is selected successfully if no
  other node's timer expires within a `vulnerability' window after its
  timer expiry, and so long as the sink can hear the available nodes.
  In this paper, we show that the optimal metric-to-timer mapping that
  (i)~maximizes the probability of success or (ii)~minimizes the
  average selection time subject to a minimum constraint on the
  probability of success, maps the metric into a set of discrete timer
  values.  We specify, in closed-form, the optimal scheme as a
  function of the maximum selection duration, the vulnerability
  window, and the number of nodes. An asymptotic characterization of
  the optimal scheme turns out to be elegant and insightful. For any
  probability distribution function of the metric, the optimal scheme
  is scalable, distributed, and performs much better than the popular
  inverse metric timer mapping. It even compares favorably with
  splitting-based selection, when the latter's feedback overhead is
  accounted for.
\end{abstract}

\begin{keywords}
  Selection, timer, cooperative communications, spatial diversity, multiuser
  diversity, multiple access, relays, VANET
\end{keywords}

\IEEEpeerreviewmaketitle

\section{Introduction}

Many wireless communication schemes benefit by selecting the `best'
node from the many available candidate nodes and using it for data
transmission. For example, cooperative communication systems exploit
spatial diversity and avoid synchronization problems among relays by
selecting the relay that is best suited to forward the source's
message to the
destination~\cite{nosratinia_ComMag_2004,zhao_2005_eurasip,beres_TWC_2008,hwang_2008_TWC,zhou_zhou_TWC_2008,jing_2009_TWC,huang_2008_TWC,kim_2007_Mobisys}.
Cellular systems exploit spatial diversity by making the base station
transmit to (or receive from) the mobile station that has the highest
instantaneous channel gain to (or from) the base station.  Fairness is
ensured by selecting on the basis of the channel gain divided by the
average throughput or average energy
consumed~\cite{Tse_2005,michalopoulos_2008_TWC}.  In sensor networks,
node selection helps increase network
lifetime~\cite{chen_2007_TSP,huang_2008_TWC}.  In vehicular ad-hoc
networks (VANETs), vehicle selection improves the speed of information
dissemination by ensuring that the vehicles that rebroadcast the
emergency broadcast message are far away from the source of the
message~\cite{nekovee_2007_VTC,yim_2009_ITS}. In some of these systems, a base
station or access point (which we shall generically refer to as a
`sink') can help the selection process by hearing transmissions from
candidate nodes and sending feedback.  On the other hand, in the
emergency broadcast scenario in VANETs, coordination issues make
explicit feedback from a sink infeasible.

The mechanism that physically selects the best node is, therefore, an
important component in many wireless systems. In all the above
systems, each node maintains a local suitability metric, and the
system attempts to select the node with the highest metric.
In~\cite{bletsas_jsac_2006}, an {\em inverse metric} timer-based
scheme was proposed, in which a node with metric $\mu$ sets its timer
as $c/\mu$, where $c$ is a constant, and transmits a packet when its
timer expires.  This simple solution ensures that the first node that
transmits is the best node.  In~\cite{zhao_2005_eurasip}, nodes with
channel gains above $\mu_u$ transmit at time $0$, while those with
channel gains below $\mu_l$ transmit at time $\tmax$.  In the interval
$[\mu_l,\mu_u)$, the mapping is linearly decreasing.  In general, to
ensure that the best node transmits first, the mapping is a
deterministic monotone non-increasing
function~\cite{bletsas_jsac_2006,zhao_2005_eurasip}.

The timer-based selection mechanism is attractive because of its
simplicity and its distributed nature.  It requires no feedback during
the selection process. A node only needs to include its identity in
the packet that it transmits upon timer expiry. A sink, if present,
only needs to broadcast a single message at the end of the selection
process indicating success or failure. Depending on the application,
the sink may also broadcast in the message the identity of the relay
has been selected. Consequently, timer-based selection has been used
in several systems such as cooperative relaying to find the best relay
node~\cite{bletsas_jsac_2006,beres_TWC_2008}, wireless network
coding~\cite{ding_2008_TWC} to find the best relays that will combine
the signals transmitted by multiple sources, mobile multi-hop
networks~\cite{kim_2007_Mobisys},
VANETs~\cite{nekovee_2007_VTC,yim_2009_ITS} to determine which vehicle
should rebroadcast an emergency message, wireless
LANs~\cite{hwang_2008_TWC} to enable opportunistic channel access, and
sensor networks~\cite{zhao_2005_eurasip,zhou_zhou_TWC_2008}.  It is
different from the centralized polling mechanism, in which the sink
node polls each node about its metric and then chooses the best one.
It also differs from the time-slotted distributed splitting
algorithms~\cite{qin_infocomm_2004,shah_2009_ICC} that also ensure
that the first packet that the sink successfully decodes is from the
best node. The difference lies in the extensive slot-by-slot
three-state (idle, collision, or success) feedback of the splitting
algorithm that controls which nodes transmit in the next slot.

The timer scheme works by ensuring that the best node transmits first.
However, for successful selection in practical systems, it is
necessary that no other timer expires within a time window of the
expiry of the best node's timer. This time window, the {\em
  vulnerability window}~\cite{kleinrock_1975_Tcom}, will be explained
in detail in the next section.  Selection failure occurs when two or
more packets collide at the receiver and become indecipherable, or
unequal node-to-sink propagation delays cause a packet from the best
node to not arrive first at the receiver.  One can decrease failure
rate by increasing the size of the vulnerability window or the maximum
selection duration, $\tmax$.  However, the latter is not desirable
because it reduces the time available to the selected node to transmit
data.  If the metrics depend on instantaneous channel fading gains,
such an increase also reduces the ability of the system to handle
larger Doppler spreads.

In this paper, we consider the general timer scheme in which the
metric-to-timer function is monotone non-increasing. In contrast to
the ad hoc mappings used in the literature, we determine the optimal
mapping that maximizes the probability of success or minimizes the
expected time required for selection subject to a minimum constraint
on the probability of success. The former is relevant in systems that
reserve a fixed amount of time for selection,
e.g.,~\cite{lo_VT_2009_HARQ}, while the latter is relevant in systems
that instead use the best node as soon as it is selected.

The specific contributions of the paper are the following.  We provide
a full recursive characterization of the optimal metric-to-timer
mapping function, and show that it is amenable to practical
implementation. We show that optimal timer schemes for the two previously stated
problems set the timer expiry at only a finite number of points in
time.  {\em That is, the optimal timer values are discrete.}  The number of points depends on the maximum allowed time
for selection $\tmax$ and vulnerability window. In the asymptotic
regime, in which the number of nodes, $k$, is large, we show that the
description of the optimal scheme and its analysis simplifies
remarkably, and takes an elegant and simple recursive form.  The
asymptotic regime turns out to be a good approximation even for $k$ as
small as 5.  Our results hold for all real-valued metrics with arbitrary probability distribution functions.

Compared to the inverse metric mapping, we show that the probability
that the system fails to select the best node can often be
substantially decreased by at least a factor of 2 for the same maximum
selection duration.  And, for a given probability of success, the
average number of slots of the optimal scheme is less by a factor of
two or more than that of the inverse timer scheme. We also show that
the optimal timer scheme is scalable in that its performance does not
catastrophically degrade as the number of nodes increases.

The paper is organized as follows. The system model and the general
timer-based selection scheme are described in Sec.~\ref{sec:System
  Model}. The optimal schemes are derived and analyzed in
Sec.~\ref{sec:success_prob_max} and \ref{sec:ps_constraint}.
Section~\ref{sec:results} presents numerical simulations and compares
with previously proposed schemes, and is followed by our conclusions
in Sec.~\ref{sec:conclusions}.  Mathematical proofs are relegated to
the Appendix.

\section{Timer-Based Selection: System Model and Basics}
\label{sec:System Model}

We consider a system with $k$ nodes and a sink as shown in
Figure~\ref{fig:systemTimer}. The sink represents any node that is
interested in the message transmitted by the $k$ nodes; it need not
conduct any coordinating role.  Each node $i$ possesses a suitability
metric $\mu_i$ that is known only to that specific node.  The metrics
are assumed to be independent and identically distributed across
nodes. The probability distribution is assumed to be known by all
nodes. The aim of the selection scheme is to make the sink determine
which node has the highest $\mu_i$, henceforth called the `best' node.


Each node $i$, based on its local metric $\mu_{i}$, sets a timer $T_i
= f(\mu_{i})$, where $f(.)$ is called the metric-to-timer function.
When the timer expires (at time $T_{i}$), the node immediately
transmits a packet to the sink. The packet contains the identity of
the node along with other system-specific information. As
mentioned, the timer-based selection scheme always ensures that the
timer of a node with a larger metric expires no later than that of a
node with a smaller metric.  Consequently, {\em $f(\mu)$, in general,
  is a monotone non-increasing function.}  The selection process has a
{\em maximum selection duration} $\tmax$, after which nodes do not
start a transmission.  

For the sink to successfully decode the packet sent by the best node,
the start time of any other packet must not be earlier than the start
time of the packet of the best node plus a vulnerability window
$\Delta$.  Thus, the sink can decode the packet from the best node, if
the timers of the best and second best node, denoted by $T_{(1)}$ and
$T_{(2)}$ respectively, expire such that $T_{(2)} - T_{(1)} \geq
\Delta$.  The expiry of timers of other nodes, which occurs after
$T_{(2)}$, does not matter since the sink is only interested in the
best node.

The value of $\Delta$ depends on system capabilities.  For example,
$\Delta$ typically includes the maximum propagation and detection
delays between all nodes. $\Delta$ may also include the maximum
transmission time of packets in case carrier sensing is not used, in
which case the nodes do not need to overhear other transmissions.
Carrier sensing, which is commonly used today, is beneficial as it
reduces $\Delta$ since a node, when its timer expires, will overhear
transmissions and does not transmit if it senses another transmission.
Note, however, that the timer scheme works with carrier sensing and
without. For a system with half-duplex nodes, $\Delta$ may also
include the receive-to-transmit switching time.

Henceforth, we will abuse the above general definition of $\Delta$ and
instead say that a {\em collision} occurs when the timers of the
best and the second best nodes expire within a duration $\Delta$.
Thus, the best node is selected successfully if: (i)~the timer value
of the best node, $T_{(1)}$, is smaller than or equal to $\tmax$, and
(ii)~the transmission from the best node does not suffer from a
collision.  Otherwise, the best node selection process fails.  The
selection process stops at $T_{(1)}$ or $\tmax$, whichever is earlier.

In this paper, inability to select the best node is treated as a
failure or an outage.  In fact, if a sink is available, it may respond
to a selection failure in multiple ways. For example,
it may resolve the nodes whose packets collided during
the selection process, using extra feedback.  If a sink is not available,
then repeated transmission schemes can be used to improve the overall
reliability of broadcast messages.  The details of how the system
deals with a selection failure are beyond the scope of this paper.

To study the performance of selection schemes, we measure the {\em
  probability of successful selection } and the {\em expected stop
  time of the selection scheme}. These are
clearly relevant to all systems that use selection.  They motivate the
following two different schemes to optimize the metric-to-timer
mapping:
\begin{enumerate}
\item {\bf Scheme 1:} Maximize the probability of success given a
  maximum selection duration of $\tmax$ and the number of nodes $k$.
\item {\bf Scheme 2:} Minimize the expected selection time given a
  maximum selection duration of $\tmax$ such that the probability of
  success is at least $\eta$ when $k$ nodes are present.
\end{enumerate}
A minimum requirement on the success probability is needed in the
second scheme because otherwise a trivial scheme that makes each of
its nodes set its timer to 0 would be optimal. This is undesirable
because the probability of success of such a scheme is zero when $k
\geq 2$.

We assume that all nodes know $k$, as is also assumed in the splitting
approach in~\cite{qin_infocomm_2004,shah_2009_ICC}. This can be
achieved, for example, by making the sink broadcast $k$ occasionally.
The burden of this feedback is not significant since $k$ typically
varies on a much slower time scale than, for example, the
instantaneous channel fades.  Even when the sink is not available,
nodes can estimate $k$ by overhearing packet transmissions in the
network.

To keep notation simple, we first consider the case where the
metrics are uniformly distributed over the interval $[0,1)$.
Thereafter, the results are
generalized to all real-valued metrics with arbitrary probability distribution functions.

{\em Notation:} Floor and ceil operations are denoted by $\floor{.}$
and $\ceil{.}$, respectively.  $\expect{X}$ denotes the expected value
of a random variable (RV) $X$. Using order statistics
notation~\cite[Chp.~1]{david_2003}, the node with the $i$th largest
metric is denoted by $(i)$.  Consequently, $\mu_{(1)} \geq \mu_{(2)}
\geq \cdots \geq \mu_{(k)}$ and $T_{(1)} \leq \cdots \leq T_{(k)}$.
For notational convenience, the summation $\sum_{l=l_{1}}^{l_{2}}$
equals $0$ whenever $l_{1} > l_{2}$. We use the superscript $^{*}$ to
denote an optimal value; for example, optimal value of $x$ is $x^{*}$.
$\prob{A}$ denotes the probability of an event $A$, and $\prob{A|B}$
denotes the conditional probability of $A$ given $B$.

\section{Scheme 1: Maximizing The Probability of Success Given $\tmax$}
\label{sec:success_prob_max}

The goal here is to find an optimal mapping $\mapopt(\mu)$ in the
space of all monotone non-increasing functions $f: [0,1) \rightarrow
\mathbb{R}^{+}$, that maximizes the probability of success.  The
following lemma shows that an optimal $\mapopt(\mu)$ maps the metrics
into discrete timer values. Let $N = \floor{\frac{\tmax}{\Delta}}$.
\begin{lemma}
\label{discrete_maxsuccess}
An optimal metric-to-timer mapping $\mapopt(\mu)$ that maximizes the
probability of success within a maximum time $\tmax$ maps $\mu$ into
$(N+1)$ discrete timer values $\{ 0, \Delta, 2\Delta,\ldots, N\Delta
\}$.
\end{lemma}
\begin{proof}
The proof is given in Appendix~\ref{lem:discrete_maxsuccess}.
\end{proof}

The discreteness result is intuitively in sync with the fact that time
slotted multiple access protocols are better than unslotted ones in
terms of throughput. However, there is a subtle but fundamental
distinction between our selection problem and the multiple access
problem. While slotting is better in multiple access protocols because
it reduces the vulnerability window, in our problem the vulnerability
window remains unchanged. Note that the above discrete mapping, while
optimal, need not be unique.  For example, when $N\Delta<\tmax$, the
highest timer value can be increased beyond $N\Delta$ without
affecting the probability of success. Also, any increase in the
discrete timer values that ensures that there are $(N+1)$ of them
below $\tmax$ and are spaced at least $\Delta$ apart, achieves the
same probability of success. Note also that the timer-based scheme is
different from the oft-employed RTS/CTS handshaking scheme, which
addresses the hidden terminal problem that may arise after the sink
receives the RTS packet successfully. In fact, the timer scheme may
even be used in the RTS backoff procedure to increase the success rate
of RTS packet reception.
%

{\em Implications of Lemma~\ref{discrete_maxsuccess}:} We have reduced an infinite-dimensional problem of finding
$f(\mu)$ over the space of all positive-valued monotone non-increasing
functions to one over $N+1$ real values that lie between 0 and
$\tmax$, as illustrated in Figure~\ref{fig:TimerScheme}.  To
completely characterize the optimal timer scheme, all we need to
specify is the contiguous metric intervals in $[0,1)$ that get
assigned to the timer values $0, \Delta,\ldots,N\Delta$. As shown in
Figure~\ref{fig:TimerScheme}, all nodes with metrics in the interval
$[1-\alp{0}{N},1)$, of length $\alp{0}{N}$, set their timers to 0.
Nodes with metrics in the next interval $[1-\alp{1}{N}-\alp{0}{N},
1-\alp{0}{N})$, of length $\alp{1}{N}$, set their timers to $\Delta$,
and so on. In general, nodes with metrics in the interval
$\left[1-\sum_{j=0}^{i} \alp{j}{N}, 1 - \sum_{j=0}^{i-1}
  \alp{j}{N}\right)$, of length $\alp{i}{N}$, set their timer to
$i\Delta$.  Any node with metric less than $\left(1-\sum_{j=0}^{N}
  \alp{j}{N}\right)$ does not transmit at all.  Therefore, the
probability of success is entirely a function of $N$,
$\alp{0}{N},\ldots,\alp{N}{N}$, and the number of nodes $k$. To keep
the notation simple, we do not explicitly show its dependency on $k$.
We now determine the optimal $\alpopt{j}{N}$ and fully characterize
the optimal scheme.
\begin{theorem}
\label{optimal_maxPs}
The probability of success in selecting the best node among $k$ nodes, subject to
a maximum selection duration of $\tmax$, is maximized when the timer
of a node with metric $\mu$ is
\begin{equation}
\mapopt(\mu) = \left\{\barr{ll}
i\Delta, &  1-\sum_{j=0}^{i} \alpopt{j}{N} \leq \mu < 1-\sum_{j=0}^{i-1} \alpopt{j}{N}, \ \text{for} \ 0\le i \le N \\
\tmax + \epsilon, & \text{otherwise}
\earr \right.,
\end{equation}
where $N = \floor{\frac{\tmax}{\Delta}}$ and $\epsilon$ is any arbitrary
 strictly positive real number. The $N+1$ interval
lengths $\alpopt{0}{N}, \ldots, \alpopt{N}{N}$ are recursively given
by
\begin{equation}
\label{eq:alp_maxPs}
\alpopt{j}{N} = \left\{
\barr{ll}
\frac{1-\Psopt{N-1}{k}}{k-\Psopt{N-1}{k}}, & j=0 \\
\left(1-\alpopt{0}{N}\right)\alpopt{j-1}{N-1}, & 1\le j \le N
\earr
\right.,
\end{equation}
where $\alpopt{0}{0} = 1/k$. $\Psopt{N}{k}$ is the maximum probability
of success that equals
\begin{equation}
\label{eq:nonrecursive_Ps}
\Psopt{N}{k}= k \sum_{l=0}^{N}\alpopt{l}{N}\left(1-\sum_{j=0}^{l}\alpopt{j}{N}\right)^{k-1}.
\end{equation}
\end{theorem}
\begin{proof}
  The proof is given in Appendix~\ref{proof of optimal_maxPs}.
\end{proof}

The discrete nature of the optimal scheme also makes it amenable to
practical implementation. Each node only needs to store an unwrapped
version of the above recursion in the form of a look up table that has
$N+1$ entries $\cbrac{\alpopt{0}{N},\ldots,\alpopt{N}{N}}$. The
entries are a function of $k$. A node only needs to determine the
interval its metric lies in and chooses the timer accordingly.


\subsection{Asymptotic Analysis of Optimal Scheme as $k \to \infty$ with Finite $N$ }
\label{sec:asymp_ps_max}

We now provide asymptotic expressions for the optimal timer scheme as
the number of nodes $k \to \infty$. The maximum selection duration,
$\tmax$, or equivalently $N$, is kept fixed. As we shall see, the
recursions simplify to a simple and elegant form because\ a scaled
version of the metric follows a Poisson process~\cite{wolff}.  The
asymptotic expressions are also relevant practically because, as we
shall see, they approximate well the optimal solution of
Theorem~\ref{optimal_maxPs} for $k$ as small as 5.

From~\eqn{eq:alp_maxPs}, it can be seen that $\alpopt{j}{N}$ tends to
zero as $k \to \infty$. Therefore, for node $i$, consider a scaled
metric $y_{i} = k(1-\mu_{i})$, and normalize the interval lengths to
\begin{equation}
\betopt{j}{N} = k\alpopt{j}{N}.
\end{equation}
Thus, selecting a node with highest $\mu_i$ is equivalent to selecting
the node with the lowest $y_i$.  Let $y_{(1)} \leq y_{(2)} \leq \cdots
\leq y_{(k)} $. Define the point process $M(z) \define
\sup{\left\{k\ge1:y_{(k)}\leq z\right\}}$. $M(z)$ is simply the number
of nodes whose $y_{i} = k(1 - \mu_{i})$ is less than $z$.
\begin{lemma}
The process $M(z)$ forms a Poisson process as $k \to \infty$.
\label{lem:poisson}
\end{lemma}
\begin{proof}
  The $\{y_{i}\}_{i=1}^k$ are uniformly and identically distributed in
  $(0,k]$. Thus,
  $\prob{M(z)=l}=\binom{k}{l}\left(\frac{z}{k}\right)^{l}\left(1-\frac{z}{k}\right)^{(k-l)}$,
  for $0 \le l \le l$, and tends to $e^{-z} \frac{z^{l}}{l!}$ as $k
  \to \infty$. Thus, it follows from~\cite[Chp.~2]{wolff} that $M(z)$
  forms a Poisson process with rate 1 as $k \to \infty$.
\end{proof}
This result enables the use of the {\em independent increments
  property} of Poisson processes~\cite[Chp.~2]{wolff}. Simply stated,
the property says that the number of points that occur in disjoint
intervals are independent of each other. We use it below to determine
the optimal $\betopt{j}{N}$.
\begin{theorem}
\label{asymp_maxPs}
The optimal $\betopt{j}{N}$ that maximize the probability of
success are given by
\begin{equation}
\label{eq:betaopt}
\betopt{j}{N} = \left\{
\barr{ll}
1, & j = N\\
1 - e^{-\betopt{j+1}{N}},& 0 \le j \le N-1
\earr \right..
\end{equation}
Also, the probability of success of the optimal scheme is
$\Psopt{N}{\infty} = e^{-\betopt{0}{N}}$.
%
\end{theorem}
\begin{proof}
 The proof is given in Appendix~\ref{proof of asymp_maxPs}.
\end{proof}

Theorem~\ref{asymp_maxPs} leads to the following key insights about
the optimal timer scheme.
\begin{corollary}[Scalability]
  As $k \to \infty$, the probability of success of the optimal scheme
  for any $\tmax$ is greater than or equal to $1/e$, with equality
  occurring only for $N = 0$.
\end{corollary}
\begin{proof}
  This follows because $\betopt{0}{0} = 1$, and $\betopt{0}{N} \leq 1$
  from the recursion in~\eqn{eq:betaopt}.
\end{proof}

\begin{corollary}[Monotonicity]
\label{cor:monotonicity}
$\betopt{0}{N} < \betopt{1}{N} < \cdots <
\betopt{N}{N}$.
\end{corollary}
\begin{proof}
  The result follows from~\eqn{eq:betaopt} and the inequality $1 -
  e^{-x} < x$, for $x > 0$.
\end{proof}
This result reflects a behavior typical of finite horizon dynamic
programming problems.  In our problem, selection at a discrete time
value does not happen when either a collision occurs or the best node
does not transmit.  As the time available decreases, the risk of
selection failure due to the best node not transmitting increases. To
counteract this, the monotonicity property makes the optimal scheme
take on a higher risk of collision.

\begin{corollary}[Independence]
 $\betopt{N-r}{N}$ depends only on $r$, and is  independent of $N$.
\end{corollary}

\begin{proof}
  Since $\betopt{N}{N} = \betopt{N-1}{N-1} = 1$, it follows
  from~\eqn{eq:betaopt} that $\betopt{j}{N} = \betopt{j-1}{N-1}$, for
  $j \ge 1$.  This also follows from the independent increments
  property: given that no node exists with $y_i \in (0,\betopt{0}{N}
  ]$, $M(z+\betopt{0}{N})$ is a rate 1 Poisson process, and time
  $(N-1)\Delta$ is left to select the best node. By arguing
  successively, we get $\betopt{j}{N} = \betopt{j-l}{N-l}$ for $j \ge
  l$.
\end{proof}


\newcommand{\lag}[2]{L_{#2}^{#1}}
\newcommand{\lagopt}[2]{L_{#2}^{*{#1}}}

\section{Scheme 2: Minimizing the Expected Selection Time}
\label{sec:ps_constraint}

Our aim now is to minimize the expected selection time,
$\avetime{N}{k}$, subject to the constraint that the probability of
success, $\Ps{N}{k}$, exceeds $\eta$.\footnote{The inclusion of the
  subscript $N$ in the symbol for probability of success will be
  become clear from Lemma~\ref{discrete_mintime}. This is done to keep
  the notation consistent throughput the paper.}  Formally, the
constrained optimization problem is:

\begin{equation}\label{eq:constrained}
\min_{f(\mu): [0,1) \rightarrow \mathbb{R}^{+}}
\avetime{N}{k} \quad \text{subject to}  \ \Ps{N}{k} \ge \eta.   
\end{equation}

{\em Feasibility of Solution}: For a given $\tmax$, a solution to the
problem above exists {\em if and only if} $\eta$ is less than or equal
to the optimum probability of success for Scheme 1.  This is because
Scheme 1, by definition, achieves the highest probability of success
given $\tmax$.  Henceforth, we shall assume that the solution is
feasible.  The following Lemma shows that the optimal mapping
$\mapopt(\mu)$ for this problem also is discrete.
\begin{lemma}
\label{discrete_mintime}
The optimal metric-to-timer mapping $\mapopt(\mu)$ that minimizes the
expected selection time subject to a minimum probability
of success constraint, $\eta$,
maps $\mu$ into $(N+1)$ discrete timer values $\{ 0, \Delta,
2\Delta,\ldots, N\Delta \}$, where $N = \floor{\frac{\tmax}{\Delta}}$.
\end{lemma}
\begin{proof}
The proof is given in Appendix~\ref{proof of discrete_mintime}.
\end{proof}

Hence, to determine the optimal mapping, it is sufficient to look at
mappings defined by the $N+1$ variables
$\alp{0}{N},\ldots,\alp{N}{N}$, where a node with metric in the
interval $[1-\alp{0}{N},0)$ sets its timer to $0$, and a node with
metric in the interval $[1-\alp{1}{N}-\alp{0}{N}, 1-\alp{0}{N})$ set
its timer to $\Delta$, and so on, as illustrated in
Figure~\ref{fig:TimerScheme}. In general, a node with metric in the
interval $[1-\sum_{j=0}^{i} \alp{j}{N}, 1 - \sum_{j=0}^{i-1}
\alp{j}{N})$, of length $\alp{i}{N}$, set its timer to $i\Delta$. A
node whose metric is less than $1-\sum_{j=0}^{N} \alp{j}{N}$ does not
transmit at all.


Consider the minimization of an auxiliary function $\lag{\lambda}{N} \define
\avetime{N}{k} - \lambda \Ps{N}{k}$, for a given $\lambda \ge 0$. We
now show that the solution that minimizes $\lag{\lambda}{N}$ is the
solution of the optimization problem in~\eqn{eq:constrained}, and that
the inequality becomes an equality.  Let $\mapoptLag(\mu)$ be the
mapping with the lowest value of the auxiliary function for a given $\lambda$. Let its
probability of success and expected selection time be $P_{N}^{\lambda
  *}$ and $\Gamma_{N}^{\lambda *}$, respectively.  Consider any other
feasible scheme $f^{\prime}(\mu)$ with corresponding probability of
success $P_{N}^{\prime}$ and expected selection time
$\Gamma_{N}^{\prime}$.  Therefore, $\Gamma_{N}^{\prime} - \lambda
P_{N}^{\prime} \geq \Gamma_{N}^{\lambda *} - \lambda P_{N}^{\lambda
  *}$.  If $P_{N}^{\prime} \ge P_{N}^{\lambda *}$, then
$\Gamma_{N}^{\lambda *} \le \Gamma_{N}^{\prime} - \lambda
\left(P_{N}^{\prime} - P_{N}^{\lambda *} \right) \le
\Gamma_{N}^{\prime}$ since $\lambda \geq 0$.  Therefore, the expected
selection time of $\mapoptLag(\mu)$ is the lowest among all timer
mappings for which $P_{N}^{\prime} \ge P_{N}^{\lambda *}$.
Consequently, if we choose $\lambda$ such that $P_{N}^{\lambda *} =
\eta$, then the resulting mapping $\mapoptLag(\mu)$ is the solution
of~\eqn{eq:constrained}.

The following theorem
specifies the optimal timer scheme as a function of $\lambda$.
\begin{theorem}
\label{optimal_mintime}

Given $\lambda \ge 0$, the auxiliary function $\lag{\lambda}{N}$ is minimized
when a node with metric $\mu$ sets its timer as $\mapopt(\mu)$, where
\begin{equation}
\mapopt(\mu) = \left\{
\barr{ll}
i\Delta, & 1-\sum_{j=0}^{i} \alpopt{j}{N} \leq \mu < 1-\sum_{j=0}^{i-1} \alpopt{j}{N}, \ \text{for} \ 0\le i \le N \\
\tmax + \epsilon, & \text{otherwise}
\earr \right.,
\end{equation}
where $N = \floor{\frac{\tmax}{\Delta}}$, $\epsilon$ is an arbitrary
positive real number.  $\alpopt{0}{N}, \ldots, \alpopt{N}{N}$ are
recursively given by
\begin{equation}
\label{eq:alp_const_Ps}
\alpopt{j}{N} = \left\{
\barr{ll}
\frac{1 + \frac{\lambda}{\Delta} - \frac{\lambda}{\Delta}\lagopt{\lambda}{N-1}}{ 1 + \frac{\lambda}{\Delta}k - \frac{\lambda}{\Delta} \lagopt{\lambda}{N-1}}, & j=0 \\
\left(1-\alpopt{0}{N}\right)\alpopt{j-1}{N-1}, & 1\le j \le N
\earr
\right.,
\end{equation}
and $\alpopt{0}{0} = 1/k$.  $\lagopt{\lambda}{N}$ is the minimum
value of the auxiliary function that equals
\begin{equation}
\lagopt{\lambda}{N} = \Delta \sum_{l=0}^{N-1}\left(1-\sum_{j=0}^{l}\alpopt{j}{N}\right)^{k} - \lambda\sum_{l=0}^{N}\alpopt{l}{N}\left(1-\sum_{j=0}^{l}\alpopt{j}{N}\right)^{k-1}.
\end{equation}

%

\end{theorem}
\begin{proof}
The proof is given in Appendix~\ref{proof of optimal_mintime}.
\end{proof}

Note that as $\lambda/\Delta \to \infty$, $\alpopt{j}{N}$ tends to the
corresponding optimal value for Scheme~1.  This is intuitive because,
for large $\lambda$, minimizing $\lag{\lambda}{N}$ is equivalent to
maximizing $\Ps{N}{k}$.  Notice also that $\alpopt{j}{N}$ and
$\Psopt{N}{k}$ depend on $\lambda$ only through the term
$\lambda/\Delta$.  Thus, as expected, the optimal solution,
$\alpopt{j}{N}$, for the constrained problem in~\eqn{eq:constrained}
does not depend on $\Delta$ for a given $N$; scaling $\Delta$ will
accordingly scale the value of $\lambda$ to ensure $\Psopt{N}{k} =
\eta$.

\subsection{Asymptotic Behavior of Optimal Scheme as $k \to \infty$ Given $N$}

We now develop an asymptotic analysis of the optimal timer scheme when
$k \to \infty$.  Define the normalized interval lengths: $\betopt{j}{N}
= k\alpopt{j}{N}$, for $j = 0,\ldots,N$.  Then, the optimal
$\betopt{j}{N}$ are as follows.

\begin{theorem}
\label{asymp_mintime}
Given $\lambda \ge 0$, the optimal values $\betopt{j}{N}$ that minimize
the auxiliary function are given by the recursion
\begin{equation}
\label{eq:beta_mintime}
\betopt{j}{N} = \left\{\barr{ll}
1, & j = N\\
1-e^{-\bet{j+1}{N}} + \Delta/\lambda, &  0 \le j \le N-1
\earr \right..
\end{equation}
The optimum probability of success is
$\Psopt{N}{\infty}=\sum_{l=0}^{N}\betopt{l}{N}e^{
  -\sum_{j=0}^{l}\betopt{j}{N} }$ and the expected selection time is
$\avetimeopt{N}{k}= \Delta \sum_{l=0}^{N-1}e^{-\sum_{j=0}^{l}
  \betopt{j}{N} }$.
\end{theorem}
\begin{proof}
  The proof is given in Appendix~\ref{proof of asymp_mintime}.  It
  also uses Lemma~\ref{lem:poisson}, which showed that the point
  process $M(z) = \sup{\left\{k\ge1: y_{(k)}\leq z\right\}}$ is a
  Poisson process with rate 1 as $k \to \infty$. Recall that
  $y_i = k(1 - \mu_i)$ and $y_{(a)} \leq y_{(b)}$, for $a \leq b$.
\end{proof}

For both the schemes, we have $\betopt{N}{N} = 1$. However,
$\betopt{j}{N}$, for $0 \le j < N$, for Scheme~2 is always greater
than or equal to that for Scheme~1. This is because of the additional
$\Delta/\lambda$ term in~\eqn{eq:beta_mintime}, which increases
$\betopt{j}{N}$.  By decreasing $\lambda$, the expected selection time
decreases and so does the probability of success.  For a given $\lambda$, it can be verified
that $\betopt{j}{N}$ satisfy the Independence property described in
Sec.~\ref{sec:success_prob_max} for Scheme 1.

\subsection{Generalization to Real-Valued Metrics with Arbitrary Probability Distributions}
\label{sec:general_cdf}

We now generalize the optimal solutions of Schemes 1 and 2 to the
general case where the metric is not uniformly distributed.  Let the
cumulative distribution function (CDF) of a metric be denoted by
$F_{c}(x) = \prob{\mu \le x}$, where $-\infty < x < \infty$.

The optimum mapping when the CDF of the metric is $F_{c}(.)$ is
$\mapopt\left(F_{c}(\mu)\right)$, where $\mapopt(.)$ is given by
Theorem~\ref{asymp_maxPs} for Scheme~1 and by
Theorem~\ref{optimal_mintime} for Scheme~2. This follows because
$F_{c}(.)$ is a monotonically non-decreasing function, and the random
variable $Y = F_{c}(\mu)$ is uniformly distributed between 0 and
1.\footnote{The CDF needs to be continuous to ensure this.  The case
  where the CDF is not continuous can be easily handled by a technique
  analogous to {\em proportional expansion} that was proposed
  in~\cite{shah_2009_TWC_Splitting} for splitting algorithms. In it, 
  each node generates a new continuous metric such that at least one
  of the nodes with the highest metric still remains the best node.}
The problem has, therefore, been reduced to the one considered earlier.
This also shows that the performance for the optimal mapping for the
two schemes does not depend on $F_{c}(.)$. Note here that we assume
that the nodes know $F_{c}(.)$. This is also assumed in the splitting
approaches~\cite{qin_infocomm_2004,shah_2009_ICC}.  Practically, this
is justified because $F_{c}(.)$, being a statistical property, can be
computed over time.

\section{Results and Performance Evaluation}
\label{sec:results}

We now study the structure and performance of the optimum timer
schemes.  We also compare them with the popular inverse metric timer
mapping that uses $f(\mu) =
c/\mu$~\cite{bletsas_jsac_2006,beres_TWC_2008,zhou_zhou_TWC_2008,ding_2008_TWC}.\footnote{A
  fair comparison with the piece-wise linear mapping
  of~\cite{zhao_2005_eurasip} is not feasible since its performance
  needs to be numerically optimized over at least 2 parameters.}  In
order to ensure a fair comparison with Schemes 1 and 2, for each
$\tmax$ and $k$, $c$ is numerically optimized to maximize the
probability of success for Scheme 1 or minimize the expected selection
time for Scheme 2. Unlike the optimal timer scheme, the performance of
the inverse metric mapping clearly depends on the probability
distribution of the metric.  For this, we shall consider a unit mean
Rayleigh distribution (with CDF $F_c(\mu) = 1 - e^{-\mu^2/2}$), which
characterizes the receive power distribution in wireless channel, and
a unit mean exponential distribution (with CDF $ F_c(\mu)=1 -
e^{-\mu}$), which is simply the square of a Rayleigh RV.

Figure~\ref{fig:MaxPsCompare} plots the maximum probability of success of
Scheme 1 ($\Psopt{N}{k}$) as a function of $N$.  Also plotted are
results from Monte Carlo simulations, which match well with the
analytical results.  It can be seen that the asymptotic curve is close
to the actual curve for $k \geq 5$. The asymptotic curve shows a
rather remarkable result: {\em regardless of $k$ and without the use
  of any feedback, the best node gets selected with a probability of
  over 75\% when $N$ is just 5.  When $N$ increases to 17, the success
  probability exceeds 90\%!}

We also see that the optimal scheme significantly outperforms the
inverse metric mapping, despite the latter's parameters being
optimized.  For example, for $N = 10$ and $N = 30$, the probability
that the system fails to select the best node for the inverse timer
scheme is respectively 2.3 and 2.5 times greater than that of the
optimum scheme for the exponential CDF. The factors increase to 2.9 and 3.2 for the Rayleigh CDF.  Thus, even though the
exponential RV is the square of the Rayleigh RV and the squaring
operation preserves the metric order, the performance of the inverse
timer scheme changes.

The structure of the optimal Scheme 1 is studied in
Figure~\ref{fig:AlpMaxPsN10}, which plots $\alpopt{j}{N}$ for $N=10$
when the metric is uniformly distributed between 0 and 1. (The
parameters for arbitrary distributions can be obtained using
Sec.~\ref{sec:general_cdf}.)  We see that $\alpopt{j}{N}$ increases
with $j$, which is in line with the asymptotic monotonicity property
of Corollary~\ref{cor:monotonicity}.

Figure~\ref{fig:MinslotsN100} considers Scheme 2 and plots the optimal
expected selection time as a function of the constraint on the
probability of success $\eta$ for $N=100$. We again see a good match
between the analytical results and the results from Monte Carlo
simulations. As in Scheme 1, the optimal scheme significantly
outperforms the optimized inverse metric mapping. For example, for
$\eta = 0.7$ and $k=5$, the optimal scheme is 5.1 and 9.6 times faster
than the optimized inverse metric mapping for the exponential and
Rayleigh CDFs, respectively. Again, the inverse metric mapping is
sensitive to the metric's probability distribution.

We now study the structure of the optimal Scheme 2.
Figure~\ref{fig:AlpMinSlotsN10k5} shows the effect of the minimum
success probability, $\eta$, on $\alpopt{j}{N}$ when $N=10$ and $k=5$.
When $\eta$ is low, the optimal timer scheme becomes faster by
tolerating a higher degree of selection failure. It maps relatively
large intervals into small timer values, and is very aggressive in the
beginning.  Also, only a small fraction of nodes do not transmit
before $\tmax$. For example, for $\eta = 0.6$ and $N=10$, only 10.7\%
of nodes have timer values greater than $\tmax$.  This result is
relevant in a high mobility environment where selection needs to be
fast as the metric values become outdated quickly.  As $\eta$
increases, the scheme becomes conservative in order to improve its
probability of success. For example, when $\eta = 0.87$ and $N = 10$,
37.5\% of nodes, on average, do not transmit at all. As $\eta$
approaches the maximum success probability of Scheme 1, Scheme 2's
parameters converge to those of Scheme~1.

\subsection{Comparison with the Splitting-Based Selection Algorithm with Feedback}

It is instructive to compare the optimal timer scheme with the
time-slotted splitting algorithm of~\cite{qin_infocomm_2004,
  shah_2009_TWC_Splitting} given that they both achieve the same goal
but in a vastly different manner. The splitting algorithm is fast; it
selects the best node within 2.467 time-slots, on average, even when
$k$ is large.  In it, the sink broadcasts feedback to the nodes at the end of every slot to specify
whether the outcome of the transmission in the slot was an idle, a success, or a collision.

We consider, as an example, selection in IEEE 802.11 wireless local
area networks~\cite{wlan_part11} that use half-duplex nodes with
carrier sensing capability. For the splitting algorithm, the duration
of a slot, in 802.11 terminology, is $2(\text{aSIFSTime} +
\text{aPreambleLength} + \text{aPLCPHeaderLength})$, where the last
two terms account for a packet's preamble and header.\footnote{Note
  that even this is optimistic because it does not account for the MAC
  packet data unit payload.} This is because each slot contains two
transmissions, the first by one or more nodes and the second by the
sink to send the 2-bit feedback (plus preambles and headers), and
every transmission is followed by a small interframe space (SIFS).  On
the other hand, the timer algorithm requires no feedback transmission.
Therefore, $\Delta = \text{aSlotTime}$, and the optimal timer scheme's
average selection time is $\avetimeopt{N}{k}$.\footnote{In case the
  system requires the sink to send a feedback message at the end of
  selection phase, this changes to $\avetimeopt{N}{k} +
  2*\text{aSIFSTime} + \text{aPreambleLength} +
  \text{aPLCPHeaderLength}$.} From~\cite[Table~17-15]{wlan_part11},
for an Orthogonal Frequency Division Multiplexing (OFDM) system with a
bandwidth of 10~MHz, $\text{aSlotTime} = 13~\mu$s, $\text{aSIFSTime} =
\text{aPreambleLength} = 32~\mu$s, and $\text{aPLCPHeaderLength} =
8~\mu$s.  Hence, the splitting algorithm's slot duration is
$144~\mu$s, which is 11 times the vulnerability window, $\Delta =
13~\mu$s, of the timer scheme.

Table~\ref{tbl:compare} shows the average selection time as a function of the probability of success constraint for the timer scheme, and compares it with splitting scheme for large $k$. Note that the splitting scheme's  probability of success is entirely determined by $\tmax$ and is not tunable.
When $\tmax$ is small, the timer scheme is faster and  can also achieve
a higher probability of success if required.  For larger $\tmax$, the
probability of success of the splitting algorithm increases
considerably; but, the timer scheme is still faster than the splitting
scheme unless the probability of success required is high.

\section{Conclusions}
\label{sec:conclusions}

We considered timer-based selection schemes that work by ensuring that
the best node's timer expires first.  Each node maps its priority
metric to a timer value, and begins its transmission after the timer
expires. We developed optimal schemes that (i)~maximized the
probability of successful selection, or (ii)~minimized the expected
selection time given a lower constraint on the probability of
successful selection.  Both the optimal schemes mapped the metrics
into $N+1$ discrete timer values, where $N = \floor{\tmax/\Delta}$. The
first scheme that maximized the probability of success also served as
a feasibility criterion for the second scheme.

We saw that a smaller vulnerability window $\Delta$ or a larger
maximum selection duration $\tmax$ improved the performance of both
schemes. In the asymptotic regime, where the number of nodes is large,
the occurrence of a Poisson process led to a considerably simpler
recursive characterization of the optimal mapping.  The optimal
schemes' performance was significantly better than the inverse metric
mapping. Unlike the latter mapping, the optimal one's performance did
not depend on the probability distribution function of the metric.

The optimal timer scheme even compared favorably with the
splitting-based selection algorithm, especially when the time
available for selection is small. This was because the slot interval
in splitting needs to include two transmissions, one from the nodes
and a feedback from the sink, and the respective switching,
propagation, and processing delays.


\appendix

\subsection{Proof of Lemma~\ref{discrete_maxsuccess}}
\label{lem:discrete_maxsuccess}

{\em The key idea behind the proof is to successively refine $f(.)$,
  by making parts of it discrete, and show that this can only improve
  the probability of success.}  Consider an arbitrary monotone
non-increasing metric-to-timer mapping $f(\mu)$.  If $\tmax < \Delta$
(\ie, $N=0$), consider the modified mapping $f_{0}(\mu)$ such that
$f_{0}(\mu) = 0$, $0 \leq f(\mu) \le \tmax$. It sets all timer values
that were less than or equal to $\tmax$ in $f(\mu)$ to 0. This does not change the
probability of success because the probability that exactly one timer
expires in the interval $[0,\tmax]$ remains the same.

When $\tmax \geq \Delta$, consider the modified mapping $f_{1}(\mu)$
derived from $f(\mu)$ as follows:
\begin{equation}
f_{1}(\mu) = \left\{
\barr{ll}
0, & 0 \le f(\mu) < \Delta \\
f(\mu), & \text{else}
\earr\right..
\end{equation}
It is easy to verify that $f_{1}(.)$ is also monotone non-increasing.
{\em We now show that the probability of success of the mapping
  $f_{1}(.)$ is always greater than or equal to that of $f(\mu)$.}

The probability of success in selecting the best node, which we denote
by $\Ps{N}{k}$, can be written as:\footnote{The subscript $N$ is used
  to maintain the same notation throughout the paper, and follows from
  the discreteness result proved in this lemma.}  $\Ps{N}{k} =
\prob{T_{(1)} \leq \tmax < T_{(2)}} + \prob{T_{(1)} \leq T_{(2)} \leq
  \tmax, T_{(2)} - T_{(1)} \geq \Delta}$.  The second term can be
further split into three mutually exclusive events: (i)~$0 \leq
T_{(1)} < \Delta \leq T_{(2)} \leq \tmax$, (ii)~$0 < \Delta \leq
T_{(1)} \le T_{(2)} \leq \tmax$, and (iii)~$0 \leq T_{(1)} \le T_{(2)}
< \Delta \leq \tmax$.  The last event does not contribute to
$\Ps{N}{k}$ as a collision will surely occur.  Therefore,
\begin{multline}
\Ps{N}{k} = \prob{T_{(1)} \leq \tmax, T_{(2)} > \tmax} + \prob{0 \leq T_{(1)} < \Delta \le T_{(2)} \leq \tmax, T_{(2)} - T_{(1)} \geq \Delta}\\
\mbox{} + \prob{0 < \Delta \le T_{(1)} \le T_{(2)} \leq \tmax, T_{(2)} - T_{(1)} \geq \Delta}.
\label{eq:ps_split}
\end{multline}

The first and third terms in~\eqn{eq:ps_split} are clearly the same
for both $f(.)$ and $f_{1}(.)$. The second term in~\eqn{eq:ps_split}
can only increase for $f_{1}(.)$ because the event $T_{(2)} - T_{(1)}
\geq \Delta$ for $f(.)$ is a subset of that of $f_{1}(.)$, and the
event $0 \leq T_{(1)} < \Delta \leq T_{(2)}\le \tmax$ is the same for
both mappings. Hence, the success probability of $f_{1}(.)$ is greater than or equal
to that of $f(.)$. Since this argument applies to any $f(\mu)$, it
also applies to the optimal $\mapopt(.)$, for which, by definition,
the probability of success cannot be increased.  The above argument is
sufficient to show the result for $\tmax < 2 \Delta$.

Otherwise, we apply an analogous argument successively as follows. Let
\begin{equation}
f_{2}(\mu) = \left\{
\barr{ll}
\Delta, & \Delta \le f_{1}(\mu) < 2\Delta\\
f_{1}(\mu), & \text{else}
\earr\right..
\end{equation}
Then, $\Ps{N}{k}$ for both $f_{1}(.)$ and $f_{2}(.)$ can be written as
\begin{multline}
  \Ps{N}{k} = \prob{T_{(1)} \leq \tmax < T_{(2)}} + \prob{0 \leq T_{(1)} \le T_{(2)} < 2 \Delta \leq \tmax, T_{(2)} - T_{(1)} \geq \Delta}\\
  \mbox{} + \prob{0 \leq T_{(1)} < 2 \Delta \le T_{(2)} \leq \tmax,
    T_{(2)} - T_{(1)} \geq \Delta} \\
 + \prob{0 < 2\Delta \le T_{(1)} \le
    T_{(2)} \leq \tmax, T_{(2)} - T_{(1)} \geq \Delta}.
\label{eq:ps_split_more}
\end{multline}
The first and fourth probability terms are clearly the same for both
mappings.  The second term is also the same for both mappings because
given that both $T_{(1)}$ and $T_{(2)}$ are less than $2 \Delta$, the
probability their difference exceeds $\Delta$ is the same for both
$f_{1}(\mu)$ and $f_{2}(\mu)$. The third probability term
in~\eqn{eq:ps_split_more} can only {\em increase} for $f_{2}(.)$
because the event $T_{(2)} - T_{(1)} \geq \Delta$ for $f_{1}(.)$ is a
subset of that of $f_{2}(.)$, and the event $0 \leq T_{(1)} \leq
2\Delta < T_{(2)} \leq \tmax$ is the same for both mappings.

Successive application of this argument shows that an optimal mapping
is discrete in the interval of $[0,N\Delta)$ and takes values in the
set $\{ 0,\Delta,2\Delta,\ldots,(N-1)\Delta \}$.  We set all $T_{i}$
in the leftover interval of $[N\Delta,\tmax]$ to $N\Delta$ without
changing the probability of success because $\tmax-N\Delta < \Delta$
and the fact that no timer value of $f_{N}(.)$ lies in the open
interval \mbox{$((N-1)\Delta,N\Delta)$}.

\subsection{Proof of Theorem~\ref{optimal_maxPs}}
\label{proof of optimal_maxPs}

In this proof, we shall denote the probability of success by
$\Ps{N}{k}(\alp{0}{N},\ldots,\alp{N}{N})$ instead of just $\Ps{N}{k}$
to show its dependence on $\alp{0}{N},\ldots,\alp{N}{N}$.  Let the
maximum probability of success, $\Psopt{N}{k}$, occur when $\alp{0}{N}
= \alpopt{0}{N},\ldots,\alp{N}{N} = \alpopt{N}{N}$.  Note that
$\alp{0}{N}+\cdots+\alp{N}{N} \leq 1$.

Given the discrete nature of the optimal timer scheme
(Lemma~\ref{discrete_maxsuccess}), success occurs at time $l\Delta$,
for $l=0,\ldots,N$, if $\mu_{(1)}$ lies in $\left[
  \left(1-\sum_{j=0}^l \alp{j}{N}\right),\left(1-\sum_{j=0}^{l-1}
    \alp{j}{N}\right)\right)$ and the remaining $k-1$ metrics lie in
$\left[0,\left(1-\sum_{j=0}^{l}\alp{j}{N}\right)\right)$. This occurs
with probability
$k\alp{l}{N}\left(1-\sum_{j=0}^{l}\alp{j}{N}\right)^{k-1}$, since the
metrics are \iid\ and uniformly distributed over $[0,1)$. Summing over
$l$ results in~\eqn{eq:nonrecursive_Ps}.

Alternately, for $N \geq 1$, the probability of success can be written
in a recursive form as follows:
\begin{multline}
\Ps{N}{k}(\alp{0}{N},\ldots,\alp{N}{N}) = \prob{ \mu_{(1)} \in [1-\alp{0}{N},1) }\prob{\text{success} | \mu_{(1)} \in [1-\alp{0}{N},1) }\\
 + \prob{\mu_{(1)} \not\in [1-\alp{0}{N},1)} \prob{\text{success} | \mu_{(1)} \not\in [1-\alp{0}{N},1)}.
\end{multline}%

Furthermore, when conditioned on $\mu_{(1)} \not\in [1-\alp{0}{N},1)$,
the $k$ metrics are \iid\ and uniformly distributed over the interval
$[0,1-\alp{0}{N})$, and $\frac{\alp{1}{N}}{1-\alp{0}{N}} + \cdots +
\frac{\alp{N}{N}}{1-\alp{0}{N}} \leq 1$.  Therefore, from the
definition of probability of success, it follows that
$\prob{\text{success} | \mu_{(1)} \not\in [1-\alp{k}{N},1)} =
\Ps{N-1}{k}\brac{\frac{\alp{1}{N}}{1-\alp{0}{N}},\ldots,\frac{\alp{N}{N}}{1-\alp{0}{N}}}$.
Hence,
\begin{align}
\Ps{N}{k}(\alp{0}{N},\ldots,\alp{N}{N}) &= k \alp{0}{N}(1-\alp{0}{N})^{k-1}\! + \!(1-\alp{0}{N})^{k} \Ps{N-1}{k}\!\!\brac{\frac{\alp{1}{N}}{1-\alp{0}{N}},\ldots,\frac{\alp{N}{N}}{1-\alp{0}{N}}}\! , \label{eq:ps_recursive}\\
 &\leq k \alp{0}{N}(1-\alp{0}{N})^{k-1} + (1-\alp{0}{N})^{k} \Psopt{N-1}{k}.
\end{align}
However, this upper bound is achieved when
$\frac{\alp{1}{N}}{1-\alpopt{0}{N}} = \alpopt{0}{N-1}, \ldots,
\frac{\alp{N}{N}}{1-\alpopt{0}{N}} = \alpopt{N-1}{N-1}$, for any given
$0 \leq \alp{0}{N} < 1$. Therefore, the maximum probability of success
given $N$ equals
\begin{align}
  \Psopt{N}{k} &=  \max_{0 \leq \alp{0}{N} < 1} \brac{k \alp{0}{N}(1-\alp{0}{N})^{k-1} + (1-\alp{0}{N})^{k} \Psopt{N-1}{k}},\label{eq:psopt_alp0N}\\
  &= k \alpopt{0}{N}(1-\alpopt{0}{N})^{k-1} + (1-\alpopt{0}{N})^{k}  \Psopt{N-1}{k},
\end{align}
where $\alpopt{0}{N}$ is the argument that
maximizes~\eqn{eq:psopt_alp0N}. Using the first order condition, we
get $\alpopt{0}{N} = \frac{1-\Psopt{N-1}{k}}{k-\Psopt{N-1}{k}}$.  For
$N = 0$, $\Psopt{0}{k} = \max_{0 \leq \alp{0}{0} \leq 1} \brac{k
  \alp{0}{0}(1-\alp{0}{0})^{k-1}}$. The maximum occurs at
$\alpopt{0}{0} = 1/k$, in which case $\Psopt{0}{k} = \brac{1 -
  1/k}^{k-1}$.

Note that the value of $\mapopt(\mu)$ when it exceeds $\tmax$ can be
left unspecified because a node does not start transmitting after
$\tmax$. This is ensured by setting $\mapopt(\mu)$ to $\tmax +
\epsilon$, where $\epsilon > 0$.

\subsection{Proof of Theorem~\ref{asymp_maxPs}}
\label{proof of asymp_maxPs}

Success occurs at time $l\Delta$ when exactly one node (the best node) has the scaled metric
$k(1 - \mu_{i})$ in the interval
$\left(\sum_{j=1}^{l-1}\bet{j}{N},\sum_{j=1}^{l}\bet{j}{N}\right]$, and
no other node has its scaled metric $k(1 - \mu_{i})$ in
$\left(0,\sum_{j=1}^{l-1}\bet{j}{N}\right)$. From the independent increments
property of Poisson  processes, selection success thus occurs with
probability $\bet{l}{N}e^{-\bet{l}{N}} \prod_{j=0}^{l-1}
e^{-\bet{j}{N}}$, which simplifies to $\bet{l}{N}
e^{-\sum_{j=0}^{l}\bet{j}{N}}$.  Summing over all $l$, we get
\begin{equation}
\label{eq:Psasymp2}
\Ps{N}{\infty}(\bet{0}{N},\ldots,\bet{N}{N}) = \sum_{l=0}^{N} \bet{l}{N}e^{ -\sum_{j=0}^{l}\bet{j}{N} }.
\end{equation}
Note that we explicitly show here the dependence of $\Ps{N}{\infty}$
on the variables $\bet{0}{N},\ldots,\bet{N}{N}$ that are being
optimized. Taking the partial derivative of $\Ps{N}{\infty}(\bet{0}{N},\ldots,\bet{N}{N})$ with
respect to $\bet{m}{N}$ and equating to 0, we get
\begin{equation}
\label{eq:deriv1}
\sum_{l=m}^{N}\betopt{l}{N}e^{- \sum_{j=m+1}^{l}\betopt{j}{N} } = 1,\quad\text{for}~ m = 0,\ldots, N,
\end{equation}
where $\betopt{m}{N}$ are the optimal values of $\bet{m}{N}$. When $m
= N$, we get $\betopt{N}{N} = 1$. For $0 \leq m \leq N-1$, upon
substituting the equation for $m+1$ into the one for $m$, we get
$\betopt{m}{N} = 1 - e^{-\betopt{m+1}{N}}$.

The optimal probability of success in~\eqn{eq:Psasymp2} can be written
as
\begin{equation}
\Psopt{N}{\infty} = e^{-\betopt{0}{N}} \sum_{l=0}^{N} \betopt{l}{N} e^{ -\sum_{j=1}^{l}\betopt{j}{N} } = e^{-\betopt{0}{N}}.
\end{equation}
The last equality follows from~\eqn{eq:deriv1}, which shows for $m = 0$ that
$\sum_{l=0}^{N} \betopt{l}{N} e^{ -\sum_{j=1}^{l}\betopt{j}{N} } = 1$.

\subsection{Proof of Lemma~\ref{discrete_mintime}}
\label{proof of discrete_mintime}

This proof also uses the successive refinement approach of
Appendix~\ref{lem:discrete_maxsuccess}. To avoid repetition, we only
highlight the main points where it differs from
Appendix~\ref{lem:discrete_maxsuccess}.

Let
$\mapopt(\mu)$ be the optimal feasible mapping.  From it,
we construct a new monotone non-increasing mapping $f_{1}(\mu)$ such
that $f_{1}(\mu) = 0$, if $0 \le \mapopt(\mu) < \Delta$, and $f_{1}(\mu) =
\mapopt(\mu)$, otherwise.  It follows from
Appendix~\ref{lem:discrete_maxsuccess} that $f_{1}(\mu)$ is also a feasible mapping since its probability of success
is greater than or equal to that of $\mapopt(.)$.
{\it Furthermore,
$f_{1}(.)$ reduces the timer values of $\mapopt(.)$
  that lie in the interval $[0,\Delta)$ to 0.  The timer values in
  $[\Delta,\tmax]$ are unchanged.  Therefore, the expected selection
  time of $f_{1}(.)$ is less than or equal to that of $\mapopt(.)$.} However, by definition of $\mapopt(.)$,
its expected selection time cannot be reduced. Applying the same argument successively, as in
Appendix~\ref{lem:discrete_maxsuccess}, we can show that the
optimal $\mapopt(\mu)$ takes only $N+1$ discrete values $0,
\Delta,\ldots, N\Delta$.

\subsection{Proof of Theorem~\ref{optimal_mintime}}
\label{proof of optimal_mintime}

\newcommand{\indic}[1]{I_{\cbrac{#1}}}

We will denote the auxiliary function as
$\lag{\lambda}{N}(\alp{0}{N},\ldots,\alp{N}{N})$ to clearly show its
dependence on $N$ and $\alp{0}{N},\ldots,\alp{N}{N}$. Similarly, the
probability of success and expected selection time are denoted by
$\Ps{N}{k}(\alp{0}{N},\ldots,\alp{N}{N})$ and
$\avetime{N}{k}(\alp{0}{N},\ldots,\alp{N}{N})$, respectively.

We first find the expression for the expected selection time. Since
$T_{(1)}/\Delta$ is an integer-valued non-negative RV
that takes values in the set $\{0,1,\ldots,N\}$, we have $
T_{(1)} = \Delta \sum_{l=0}^{N-1} \indic{T_{(1)}/\Delta > l}$, where $\indic{x}$ is an indicator function that equals 1 if condition
$x$ is true, and is 0 otherwise.  Taking expectations on both sides, we
get
\begin{align}
  \avetime{N}{k}(\alp{0}{N},\ldots,\alp{N}{N}) = \Delta \sum_{l=0}^{N-1}
  \prob{T_{(1)}/\Delta > l} = \Delta \sum_{l=0}^{N-1} \brac{1 -
    \sum_{j=0}^{l}\alp{j}{N} }^{k}.
\label{eq:avetime}
\end{align}

Alternately, $\avetime{N}{k}(\alp{0}{N},\ldots,\alp{N}{N})$ can also
be written recursively as follows. The probability of the event that no node
transmits at time $0$ is $\brac{1 - \alp{0}{N} }^{k}$. Conditioned on
this event, the $k$ metrics are \iid\ and uniformly distributed over
the interval $[0,1-\alp{0}{N})$. The nodes can now use only
the $(N-1)$ timer values in the set
$\{\Delta,2\Delta,\ldots,N\Delta\}$.  Thus, we get
\begin{multline}
\avetime{N}{k}(\alp{0}{N},\ldots,\alp{N}{N}) = 0\left( 1-\brac{1 - \alp{0}{N} }^{k} \right) + \\
\brac{1 - \alp{0}{N} }^{k} \left( \Delta + \avetime{N-1}{k}\brac{\frac{\alp{1}{N}}{1-\alp{0}{N}},\ldots,\frac{\alp{N}{N}}{1-\alp{0}{N}}} \right).
\label{eq:avetime_recursive}
\end{multline}
From the recursive forms in~\eqn{eq:avetime_recursive} and
\eqn{eq:ps_recursive}, we get
%
%
\begin{multline}
\lag{\lambda}{N}\brac{\alp{0}{N},\ldots,\alp{N}{N}} =\Delta(1-\alp{0}{N})^{k}  - \lambda k \alp{0}{N}(1-\alp{0}{N})^{k-1}\\
\mbox{} + \brac{1 - \alp{0}{N} }^{k} \lag{\lambda}{N-1}\brac{\frac{\alp{1}{N}}{1-\alp{0}{N}},\ldots,\frac{\alp{N}{N}}{1-\alp{0}{N}}}.
\end{multline}

Since $\frac{\alp{1}{N}}{1-\alp{0}{N}} + \cdots +
\frac{\alp{N}{N}}{1-\alp{0}{N}} \leq 1$, it follows from the
definition of $\lagopt{\lambda}{N}$ that
\begin{multline}
  \lag{\lambda}{N}\brac{\alp{0}{N},\ldots,\alp{N}{N}} \geq \Delta (1-\alp{0}{N})^{k}  - \lambda k \alp{0}{N}(1-\alp{0}{N})^{k-1} +
  \brac{1 - \alp{0}{N} }^{k} \lagopt{\lambda}{N-1},\nonumber
\end{multline}
with equality when
$\frac{\alp{1}{N}}{1-\alp{0}{N}} = \alpopt{0}{N-1}, \ldots,
\frac{\alp{N}{N}}{1-\alp{0}{N}} = \alpopt{N-1}{N-1}$, for any $\alp{0}{N}$. Therefore,
\begin{align}
  \lagopt{\lambda}{N} & = \min_{0 \leq \alp{0}{N} < 1}
  \brac{\Delta(1-\alp{0}{N})^{k} - \lambda k
    \alp{0}{N}(1-\alp{0}{N})^{k-1} +
    \brac{1 - \alp{0}{N} }^{k} \lagopt{\lambda}{N-1}},\nonumber\\
  &= \Delta(1-\alpopt{0}{N})^{k} - \lambda k
  \alpopt{0}{N}(1-\alpopt{0}{N})^{k-1} + \brac{1 - \alpopt{0}{N} }^{k}
  \lagopt{\lambda}{N-1}.
\end{align}
From the first order condition, we have $\alpopt{0}{N} = \frac{1+
  \frac{\lambda}{\Delta} - \frac{\lambda}{\Delta}
  \lagopt{\lambda}{N-1}}{ 1 + \frac{\lambda}{\Delta} k -
  \frac{\lambda}{\Delta} \lagopt{\lambda}{N-1}}$.
Furthermore, for $N = 0$, we have $\avetime{N}{k} = 0$. Therefore,
$\lagopt{\lambda}{0} = \min_{\alp{0}{0}} \lambda \brac{k \alp{0}{0}
  \brac{1-\alp{0}{0}}^{k-1}}$. The optimal value $\alpopt{0}{0}$ that
minimizes this expression, for any $\lambda > 0$, is $\alpopt{0}{0} =
1/k$.

\subsection{Proof of Theorem~\ref{asymp_mintime}}
\label{proof of asymp_mintime}

The expression for the $\Ps{N}{k}(\bet{0}{N},\ldots,\bet{N}{N})$
 as a function of
$\bet{j}{N}$, for $j = 0,\ldots,N$, follows directly
from~\eqn{eq:Psasymp2}. The expression for $\avetime{N}{k}(\bet{0}{N},\ldots,\bet{N}{N})$ can be written as
\begin{align}
  \avetime{N}{k}(\bet{0}{N},\ldots,\bet{N}{N}) = \Delta \sum_{l=0}^{N-1}
  \prob{T_{(1)}/\Delta > l} = \Delta \sum_{l=0}^{N-1}e^{-\sum_{j=0}^{l} \bet{j}{N} },
\end{align}
where the first equality follows from~\eqn{eq:avetime} and the last equality
follows from the Poisson process result of
Lemma~\ref{lem:poisson}.  The auxiliary function then equals
\begin{equation}
\lag{\lambda}{N}(\bet{0}{N},\ldots,\bet{N}{N}) = \Delta \sum_{l=0}^{N-1}e^{-\sum_{j=0}^{l} \bet{j}{N} } -
\lambda \left( \sum_{l=0}^{N}\bet{l}{N}e^{ -\sum_{j=0}^{l}\bet{j}{N} }
\right).
\end{equation}
From the first order condition, it follows that $\lag{\lambda}{N}$ is
minimized by $\betopt{j}{N} = 1 - e^{-\betopt{j+1}{N}} +
\frac{\Delta}{\lambda}$.

\bibliographystyle{ieeetr}
\bibliography{../../../Bibtex/bibJournalList,../../../Bibtex/database,../../../Bibtex/cooperativeComm,../../../Bibtex/lognormal,../../../Bibtex/MIMO,../../../Bibtex/mimoEstimation,../../../Bibtex/cdma,../../../Bibtex/adaptation,../../../Bibtex/scheduling,../../../Bibtex/book,../../../Bibtex/standard}

\begin{thebibliography}{10}

\bibitem{nosratinia_ComMag_2004}
A.~Nosratinia, T.~Hunter, and A.~Hedayat, ``Cooperative communication in
  wireless networks,'' {\em IEEE Commun.\ Mag.}, vol.~42, pp.~68--73, 2004.

\bibitem{zhao_2005_eurasip}
Q.~Zhao and L.~Tong, ``Opportunistic carrier sensing for energy-efficient
  information retrieval in sensor networks,'' {\em {EURASIP} J.\ Wireless
  Commun.\ and Networking}, pp.~231--241, May 2005.

\bibitem{beres_TWC_2008}
E.~Beres and R.~Adve, ``Selection cooperation in mutli-source cooperative
  networks,'' {\em IEEE Trans.\ Wireless Commun.}, vol.~7, pp.~118--127, Jan.
  2008.

\bibitem{hwang_2008_TWC}
C.~S. Hwang, K.~Seong, and J.~M. Cioffi, ``Throughput maximization by utilizing
  multi-user diversity in slow-fading random access channels,'' {\em IEEE
  Trans.\ Wireless Commun.}, vol.~7, pp.~2526--2535, Jul. 2008.

\bibitem{zhou_zhou_TWC_2008}
Z.~Zhou, S.~Zhou, J.-H. Cui, and S.~Cui, ``Energy-efficient cooperative
  communication based on power control and selective single-relay in wireless
  sensor networks,'' {\em IEEE Trans.\ Wireless Commun.}, vol.~7,
  pp.~3066--3078, Aug. 2008.

\bibitem{jing_2009_TWC}
Y.~Jing and H.~Jafarkhani, ``Single and multiple relay selection schemes and
  their achievable diversity orders,'' {\em IEEE Trans.\ Wireless Commun.},
  vol.~8, pp.~1414--1423, Mar. 2009.

\bibitem{huang_2008_TWC}
W.-J. Huang, Y.-W.~P. Hong, and C.-C.~J. Kuo, ``Lifetime maximization for
  amplify-and-forward cooperative networks,'' {\em IEEE Trans.\ Wireless
  Commun.}, vol.~7, pp.~1800--1805, May 2008.

\bibitem{kim_2007_Mobisys}
J.~Kim and S.~Bohacek, ``A comparison of opportunistic and deterministic
  forwarding in mobile multihop wireless networks,'' in {\em Proc.\ 1st Intl.\
  MobiSys Workshop Mobile Opportunistic Netw.}, pp.~9--16, Jun. 2007.

\bibitem{Tse_2005}
D.~Tse and P.~Vishwanath, {\em Fundamentals of Wireless Communications}.
\newblock Cambridge University Press, 2005.

\bibitem{michalopoulos_2008_TWC}
D.~S. Michalopoulos and G.~K. Karagiannidis, ``{PHY}-layer fairness in amplify
  and forward cooperative diversity systems,'' {\em IEEE Trans.\ Wireless
  Commun.}, vol.~7, pp.~1073--1083, Mar. 2008.

\bibitem{chen_2007_TSP}
Y.~Chen and Q.~Zhao, ``An integrated approach to energy-aware medium access for
  wireless sensor networks,'' {\em IEEE Trans.\ Signal Process.}, vol.~7,
  pp.~3429--3444, Jul. 2007.

\bibitem{nekovee_2007_VTC}
M.~Nekovee and B.~B. Bogason, ``Reliable and efficient information
  dissemination in intermittently connected vehicular adhoc networks,'' in {\em
  Proc.\ {VTC} (Spring)}, pp.~2486--2490, 2007.

\bibitem{yim_2009_ITS}
R.~Yim, J.~Guo, P.~Orlik, and J.~Zhang, ``Received power-based prioritized
  rebroadcasting for {V2V} safety message dissemination,'' in {\em Int.\
  Transport.\ Sys.\ World Congr.}, Sept. 2009.

\bibitem{bletsas_jsac_2006}
A.~Bletsas, A.~Khisti, D.~P. Reed, and A.~Lippman, ``A simple cooperative
  diversity method based on network path selection,'' {\em IEEE {J.} Sel.\
  Areas Commun.}, vol.~24, pp.~659--672, Mar. 2006.

\bibitem{ding_2008_TWC}
Z.~Ding, T.~Ratnarajah, and K.~K. Leung, ``On the study of network coded {AF}
  transmission protocol for wireless multiple access channels,'' {\em IEEE
  Trans.\ Wireless Commun.}, vol.~7, pp.~4568--4574, Nov. 2008.

\bibitem{qin_infocomm_2004}
X.~Qin and R.~Berry, ``Opportunistic splitting algorithms for wireless
  networks,'' in {\em Proc.\ INFOCOM}, pp.~1662--1672, Mar. 2004.

\bibitem{shah_2009_ICC}
V.~Shah, N.~B. Mehta, and R.~Yim, ``Analysis, insights and generalization of a
  fast decentralized relay selection mechanism,'' in {\em Proc.\ {ICC}}, 2009.

\bibitem{kleinrock_1975_Tcom}
L.~Kleinrock and F.~Tobagi, ``Packet switching in radio channels: Part {I} --
  {C}arrier sense multiple access modes and their throughput delay
  characteristics,'' {\em IEEE Trans.\ Commun.}, vol.~23, pp.~1400--1416, Dec.
  1975.

\bibitem{lo_VT_2009_HARQ}
C.~K. Lo, J.~R.~W.~Heath, and S.~Vishwanath, ``The impact of channel feedback
  on opportunistic relay selection for hybrid-{ARQ} in wireless networks,''
  {\em IEEE Trans.\ Veh.\ Technol.}, vol.~58, pp.~1255--1268, Mar. 2009.

\bibitem{david_2003}
H.~A. David and H.~N. Nagaraja, {\em Order Statistics}.
\newblock Wiley Series in Probability and Statistics, 3~ed., 2003.

\bibitem{wolff}
R.~W. Wolff, {\em Stochastic Modeling and the Theory of Queues}.
\newblock Prentice Hall, 1989.

\bibitem{shah_2009_TWC_Splitting}
V.~Shah, N.~B. Mehta, and R.~Yim, ``Splitting algorithms for fast decentralized
  cooperative relay selection,'' {\em Submitted to IEEE Trans.\ Wireless
  Commun.}, 2009.

\bibitem{wlan_part11}
``Part 11: Wireless {LAN} medium access control ({MAC}) and physical layer
  ({PHY}) specifications,'' Tech. Rep. IEEE Std 802.11-2007, IEEE Computer
  Society, June 2007.

\end{thebibliography}

\begin{figure}[p]
\centering
\input{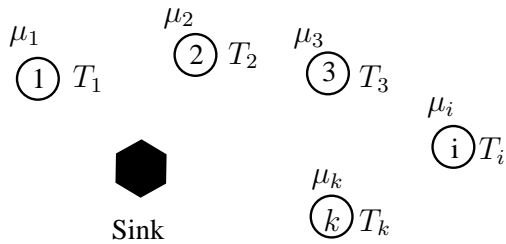}
\caption{A system consisting of a sink and $k$ nodes. Each node has a metric $\mu_i$ and sets its timer $T_i  = f(\mu_i)$. The sink needs to select the node with the higest metric.}
\label{fig:systemTimer}
\end{figure}
\begin{figure}[p]
  \centering \input{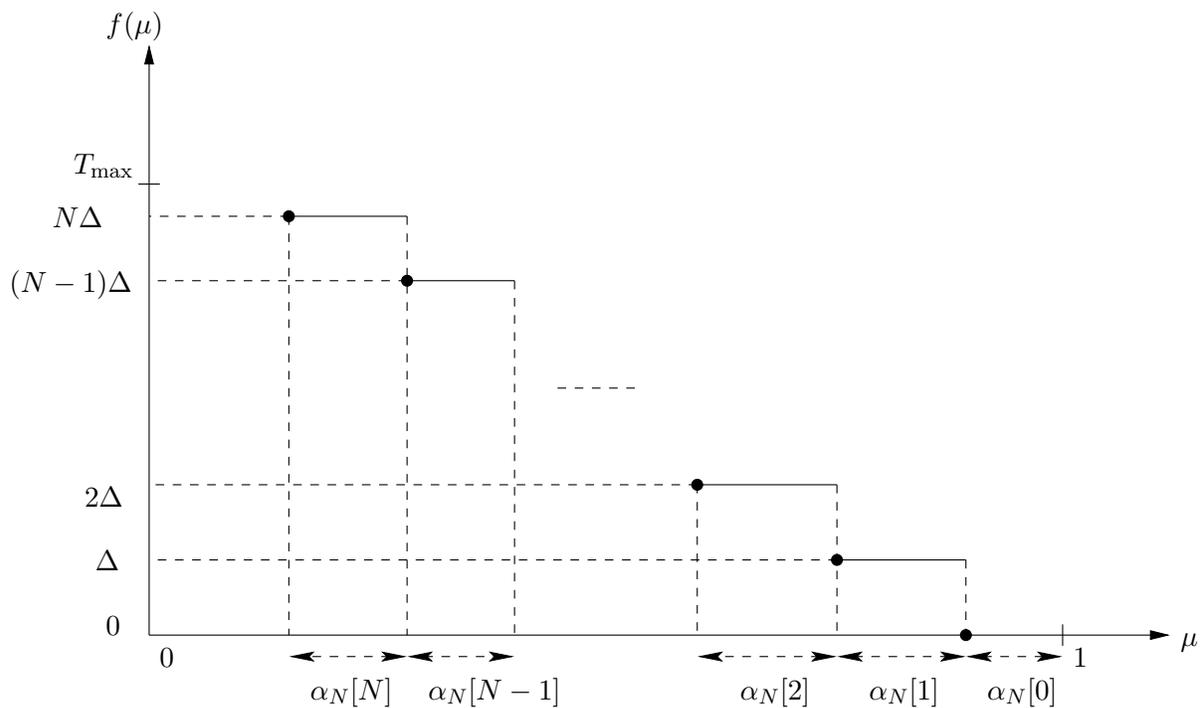}
\caption{Illustration of the optimal metric-to-timer mapping $\mapopt(\mu)$. A user with metric in the interval $[1-\alp{0}{N},1)$ transmits at time 0, a user with metric in the interval $[1-\alp{0}{N} - \alp{1}{N}, 1-\alp{0}{N})$ transmits at time $\Delta$, and so on. A user whose metric is less than $1 - \sum_{i=0}^{N}\alp{i}{N}$ does not transmit.}
\label{fig:TimerScheme}
\end{figure}
%

\begin{figure}[p]
  \centering \includegraphics[width=0.8\columnwidth,
  keepaspectratio]{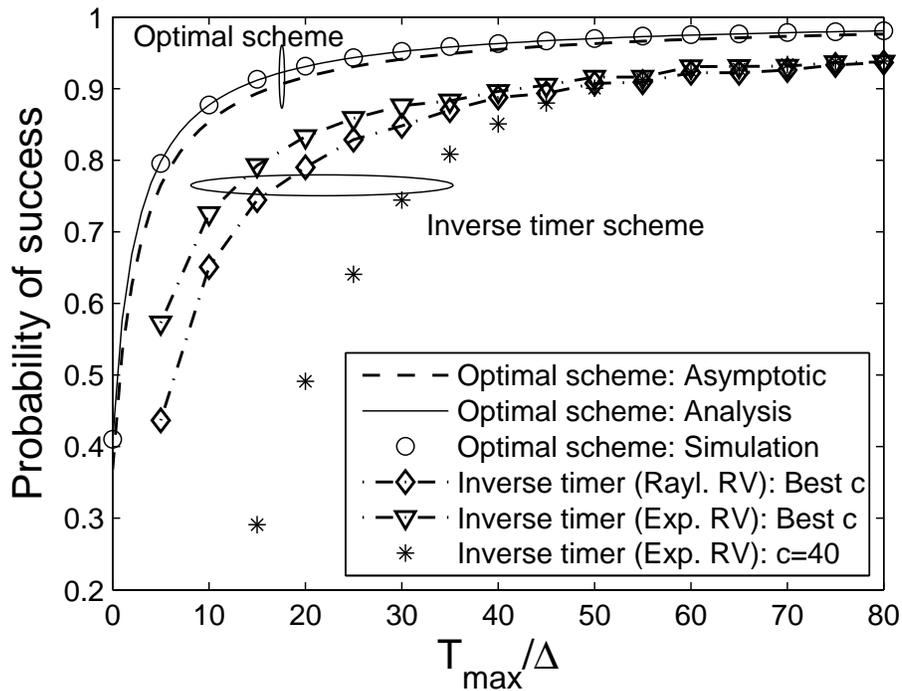}
\caption{Scheme 1: Optimum $\Psopt{N}{k}$ as a function of $\frac{\tmax}{\Delta}$. Also, plotted is the probability of success of the inverse metric mapping ($k=5$) when $c$ is optimized and when $c$ is kept fixed. }
\label{fig:MaxPsCompare}
\end{figure}

\begin{figure}[p]
  \centering \includegraphics[width=0.8\columnwidth,
  keepaspectratio]{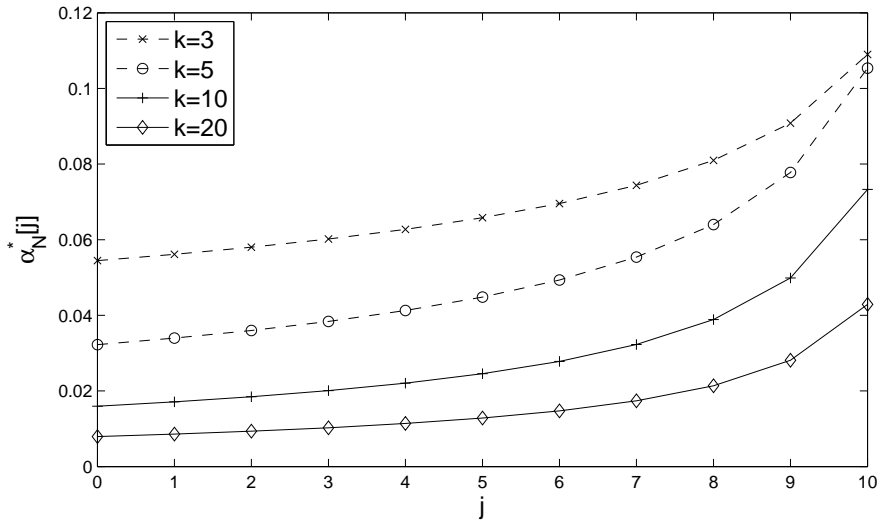}
\caption{Scheme 1: Optimum $\alpopt{j}{N}$ as a function of $j$ and the number of nodes, $k$, for $N = 10$.}
\label{fig:AlpMaxPsN10}
\end{figure}

\begin{figure}[p]
  \centering \includegraphics[width=0.8\columnwidth,
  keepaspectratio]{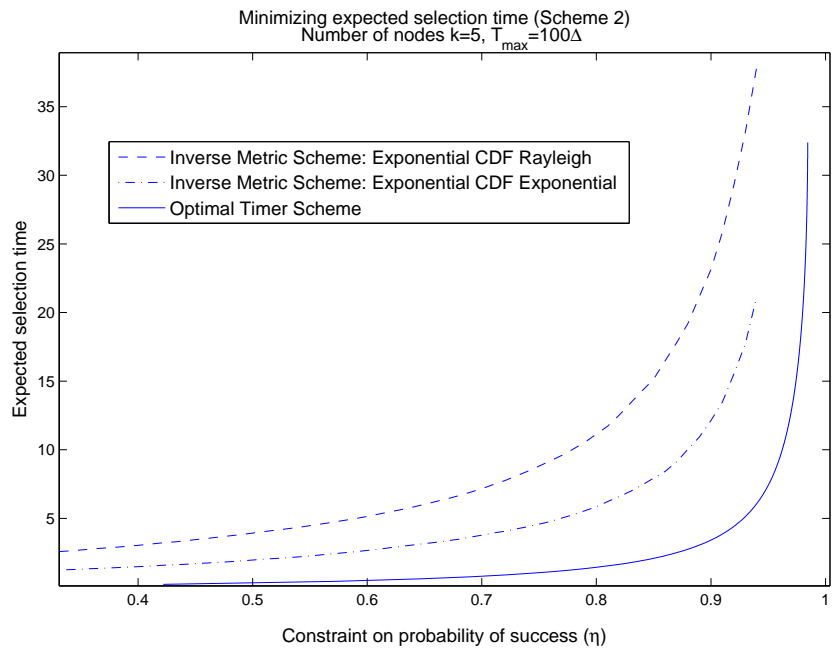}
\caption{Scheme 2: Optimum expected selection time as a function of constraint on probability of success  $(\Ps{N}{k}\le \eta)$ for $\tmax = 100\Delta$ and $k = 5$. Also, plotted is the selection time of the inverse metric mapping.}
\label{fig:MinslotsN100}
\end{figure}

\begin{figure}[p]
  \centering \includegraphics[width=0.8\columnwidth,
  keepaspectratio]{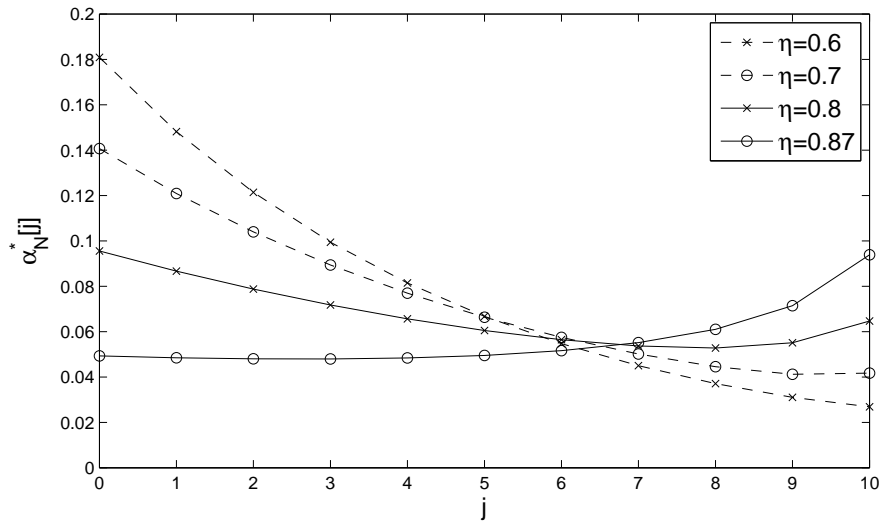} \caption{Scheme 2: Optimum
    $\alp{j}{N}$ as a function of $j$ and $\eta$ for $N = 10$ and
    $k=5$.}
\label{fig:AlpMinSlotsN10k5}
\end{figure}

\begin{table}[p]
\centering
\begin{tabular}{|l|c|c|c|c|c||c|}\hline
$\tmax$ & \multicolumn{5}{|c||}{Optimal Timer Scheme} & Splitting\\ \hline
\multirow{2}{*}{$288~\mu$s} & Prob.\ success &  0.75 & 0.85 & 0.90 & 0.98 & 0.63\\ \cline{2-7}
& Ave.\ selection time ($\mu$s) &    17.8 & 35.0 & 58.3 & -- & 233.3\\ \hline\hline
\multirow{2}{*}{$1296~\mu$s} & Prob.\ success &  0.75 & 0.85 & 0.90 & 0.98 & 0.99\\ \cline{2-7}
& Ave.\ selection time ($\mu$s) &  17.7 & 34.9 & 56.4 & 369.2 & 354.4 \\ \hline
\end{tabular}
\caption{Comparison of Optimal Timer Based Scheme 2 and Splitting Scheme}
\label{tbl:compare}
\end{table}

\end{document}